\documentclass[preprint, 3p, 
]{elsarticle} 
\usepackage{mathrsfs}
\usepackage{amsfonts}
\usepackage{amsmath}
\usepackage{amssymb}
\usepackage{amsthm}
\usepackage{txfonts}
\usepackage{graphicx}
\usepackage{bm}
\usepackage{color}
\usepackage{hyperref}
\usepackage[normalem]{ulem}

\newcommand{\ket}[1]{|#1\rangle}
\newcommand{\bra}[1]{\langle #1|}

\newtheorem{thm}{Theorem}

\begin{document}

\title{Sum rule for non-adiabatic geometric phases}

\author{Adam Fredriksson\corref{} and Erik Sj\"oqvist\corref{cor1}}

\cortext[cor1]{erik.sjoqvist@physics.uu.se} 

\address{Department of Physics and Astronomy, Uppsala University, Box 516, 
SE-751 20, Uppsala, Sweden}

\date{\today}

\begin{abstract}
Berry monopoles always cancel when summing over a complete set of energy eigenstates. 
We demonstrate that analogous sum rules exist for geometric phases and their underlying 
2-forms in non-adiabatic evolution. Our result has implications for qudit computation as it 
limits the types of gates that can be implemented by purely geometric means.
\end{abstract}

\maketitle 

Berry's discovery \cite{berry84} of the geometric phase accompanying cyclic adiabatic 
changes in quantum systems opened the door to studying magnetic monopole analogues 
in the laboratory. These monopoles have been predicted and detected in various physical 
settings, such as in the context of Bose-Einstein condensates \cite{ray15}, in Weyl 
semimetal \cite{unzelmann21}, as well as in spin Hall effect \cite{ishida23}.

Consider a Hamiltonian $H(\mathbf{R})$ with non-degenerate eigenvalues $\{ E_n ({\bf R})\}$ 
and control parameters $\mathbf{R} = (R_{1}, \ldots, R_{M})$ that evolve adiabatically along 
a closed curve $\mathcal{C}$ enclosing a surface $\mathcal{S}$. The Berry phase associated 
with the loop $\mathcal{C}$ and picked up by the eigenstate $\ket{n(\mathbf{R})}$ of $H(\mathbf{R})$ 
reads  \cite{berry84}
\begin{eqnarray}
\gamma_{n}(\mathcal{C}) = \oint_{\mathcal{C}} \mathcal{A}^{(n)} = 
\int_{\mathcal{S}|\partial \mathcal{S} = \mathcal{C}} \mathcal{F}^{(n)} .
\label{eq:BerryPhase}
\end{eqnarray}
Here, $\mathcal{A}^{(n)} \equiv i\bra{n}dn\rangle$  
and $\mathcal{F}^{(n)} \equiv d\mathcal{A}^{(n)} = i\bra{dn}\wedge\ket{dn}$
are the Berry connection 1-form and curvature 2-form \cite{simon83}, respectively, with 
exterior derivative $d=\sum_{j = 1}^{M} dR_{j} \partial /\partial R_{j}$ on parameter space. In 
Ref.~\cite{chruscinski04}, it is shown that the sum of the Berry curvatures at each point 
in parameter space vanishes, i.e.,   
\begin{equation}
\sum_{n}\mathcal{F}^{(n)} = 0 .
\label{eq:CJprop}
\end{equation}
In particular, this implies that the monopole charges 
$q_n = \int_{\mathcal{V}|\partial \mathcal{V}=\mathscr{S}} d\mathcal{F}^{(n)}$  within any closed 
surface $\mathscr{S}$ in parameter space, always cancel when summed over all states 
\cite{eriksson20}. 

Now, the Berry phase has been generalized to non-adiabatic time evolution associated with  
closed curves in projective Hilbert space \cite{aharonov87}. As the energy eigenstates in 
the case of the Berry phase may alternatively be viewed as evolving in projective Hilbert 
space, one may expect that the sum rule in Eq.~\eqref{eq:CJprop} should hold for 
non-adiabatic time evolution as well. The purpose of the present work is to show that 
this is indeed the case. 

To this end, consider a $d$ dimensional Hilbert space $\mathscr{H}_{d}$ ($d$ finite). Let 
$x_{\mu}$ be local coordinates of projective Hilbert space 
$\mathscr{P}(\mathscr{H}_{d})\cong\mathbb{C}\mathbf{P}^{d-1}\cong S^{2d-1}/S^{1}$. 
A cyclic evolution is a loop $C$ in $\mathscr{P}(\mathscr{H}_{d})$, for which one may 
associate a geometric phase \cite{aharonov87}
\begin{eqnarray}
\gamma_{g}(C) = \oint_{C}A = \int_{S|\partial S = C}F 
\label{eq:AAphase}
\end{eqnarray}
with $A \equiv i\bra{\phi}d\phi\rangle$ and $F \equiv dA = i\bra{d\phi}\wedge\ket{d\phi}$ 
the non-adiabatic connection 1-form and curvature 2-form, respectively. Here, $\ket{\phi}$ 
is a local section of the principal fiber bundle over $\mathscr{P}(\mathscr{H}_{d})$  with 
group U(1) \cite{bohm91} and corresponding exterior derivative $d=\sum_{\mu} dx_{\mu} 
\partial/\partial x_{\mu}$. 
 
A possible physical realization of $C$ is an eigenstate $\ket{\psi_0}$ of continuous unitary 
time evolution $U(t,0)=e^{-\frac{i}{\hbar} Ht}$ generated by a Hamiltonian $H$. Here, 
$\ket{\phi (s)} = e^{-if(s)} U(s,0) \ket{\psi_0}$, $s\in [0,t]$, with $f(s)$ differentiable and chosen so 
that $\ket{\phi (t)}=\ket{\phi (0)}$ (i.e., $f(t)-f(0) = \alpha$ with $U(t,0)\ket{\psi_0} = 
e^{i\alpha} \ket{\psi_0}$). Another realization is a sequence of complete projective 
measurements $P_1,\ldots,P_K$ generating a loop consisting of segments of null 
phase curves \cite{rabei99}. 

We are now ready to state and prove our first result: 
\begin{thm}
Given an orthonormal set $\{\ket{\phi_{j}}\}_{j = 1}^{d}$ of local sections spanning 
$\mathscr{H}_d$, in terms of which we define the curvature 2-forms $F^{(j)} = i\bra{d\phi_{j}} 
\wedge \ket{d\phi_{j}}$, it holds that 
\begin{eqnarray}
\sum_{j=1}^d F^{(j)} = i\sum_{j=1}^d \bra{d\phi_j} \wedge \ket{d\phi_j} = 0 .   
\label{eq:thm1}
\end{eqnarray}
\end{thm}

\begin{proof}
By using the completeness relation $\sum_{k = 1}^{d}\ket{\phi_{k}}\bra{\phi_{k}} = \hat{1}$, 
we find   
\begin{eqnarray}
\sum_{j=1}^{d} F^{(j)} = i\sum_{j,k=1}^d \langle d\phi_j \ket{\phi_k} \wedge \langle \phi_k \ket{d\phi_j} .  
\end{eqnarray}
The anti-symmetry of the $\wedge$ product and $\langle \phi_{j'} \ket{d\phi_{k'}} = 
- \langle d\phi_{j'} \ket{\phi_{k'}}$ for any pair $j',k'$, yield 
\begin{eqnarray}
\sum_{j=1}^d F^{(j)} & = & \frac{i}{2} \sum_{j,k=1}^d \Big( \langle d\phi_j \ket{\phi_k} \wedge 
\langle \phi_k \ket{d\phi_j} - \langle \phi_k \ket{d\phi_j} \wedge \langle d\phi_j \ket{\phi_k} \Big) 
\nonumber \\ 
 & = & \frac{i}{2} \sum_{j,k=1}^d \Big( \langle d\phi_j \ket{\phi_k} \wedge 
 \langle \phi_k \ket{d\phi_j} - \langle d\phi_k \ket{\phi_j} \wedge \langle \phi_j \ket{d\phi_k} \Big) = 0, 
\end{eqnarray}
where the last step follows from switching $j \leftrightarrow k$ in the second term inside the sum. 
\end{proof}

It is illustrative to verify Eq.~\eqref{eq:thm1} explicitly for the case of a qubit, i.e., for $d = 2$. 
We may write $\ket{\phi_{1}} = \mathcal{N} \big( \ket{0} + x \ket{1} \big)$ and 
$\ket{\phi_{2}} = \mathcal{N} \,  \big( -x^{\ast} \ket{0} + \ket{1} \big)$ with $x \in \mathbb{C}$ and 
$\mathcal{N} = (1+|x|^2)^{-1/2}$ \cite{page87}. By using the anti-symmetry of $\wedge$, one finds: 
\begin{eqnarray}
\bra{d\phi_{1}} \wedge \ket{d\phi_{1}} + \bra{d\phi_{2}} \wedge \ket{d\phi_{2}} 
 & = & 2d\mathcal{N} \wedge d\mathcal{N} \big( 1+ |x|^2\big) + 
 x^{\ast}\mathcal{N}  \left( d\mathcal{N}  \wedge dx + dx \wedge d\mathcal{N} \right) 
 \nonumber \\ 
  & & + x\mathcal{N}  \left( dx^{\ast} \wedge d\mathcal{N}  + 
 d\mathcal{N}  \wedge dx^{\ast} \right) + 
 \mathcal{N}^2 \left( dx^{\ast}  \wedge dx + dx \wedge dx^{\ast} \right) = 0 . 
\end{eqnarray}

We can now formulate the following theorem being the non-adiabatic analogue of the 
above mentioned cancellation of monopole charge in the adiabatic case: 
\begin{thm}
Let 
\begin{eqnarray}
\gamma_{g}^{(j)} = \int_{S_j|\partial S_j = C_j} F^{(j)}
\end{eqnarray}
be the geometric phase associated with the state $j$. These phases satisfy the sum rule: 
\begin{eqnarray}
\prod_{j=1}^d e^{i\gamma_{g}^{(j)}} = e^{i\sum_{j=1}^d \gamma_{g}^{(j)}} = 1. 
\label{eq:gp_rule}
\end{eqnarray}
\end{thm}
\begin{proof}
The theorem is immediate from Theorem 1 by noting that there exist local coordinates 
$x_{\mu}$ and a loop $C:s\in [0,t] \mapsto x_{\mu}$ whose projections are the $C_j$:s. 
Thus, 
\begin{eqnarray}
\sum_{j=1}^d \gamma_{g}^{(j)} = \sum_{j=1}^d \int_{S_j|\partial S_j = C_j} F^{(j)} = 
\int_{S|\partial S = C} \sum_{j=1}^d F^{(j)} = 0, 
\end{eqnarray}
which implies Eq.~\eqref{eq:gp_rule}. 
\end{proof}

An explicit construction of $x_{\mu}$ are as local coordinates of a manifold spanned by 
probability amplitudes $y_{\mu}^{(j)}$ of the local sections $\{ \ket{\phi_j} = 
\sum_{\mu = 0}^{d-1} y_{\mu}^{(j)} \ket{\mu} \}$ introduced in Theorem 1 above, and 
constrained by the orthonormality conditions 
\begin{eqnarray}
\sum_{\mu = 0}^{d-1} y_{\mu}^{(j)\ast} y_{\mu}^{(k)} = \delta_{jk}.  
\label{eq:orthonormality}
\end{eqnarray}
For a single qubit, a map of $\ket{\phi_j} = y_0^{(j)} \ket{0} + y_1^{(j)} \ket{1}$, $j=1,2$, to 
the local coordinate $x \in \mathbb{C}$ that satisfies Eq.~\eqref{eq:orthonormality} is 
$(y_0^{(1)},y_1^{(1)}) \mapsto \mathcal{N} (1,x)$ and $(y_0^{(2)},y_1^{(2)}) \mapsto 
\mathcal{N} (-x^{\ast},1)$. The two stereographically projected loops $C_1$ and $C_2$ 
can be viewed as mirror images on the Bloch sphere. The corresponding geometric phase 
factors are therefore  $e^{\mp i \frac{1}{2}\Omega}$ with $\Omega$ the solid angle 
subtended by $C_1$ (say), which verify Eq.~\eqref{eq:gp_rule}.

More generally, Theorem 2 entails that no evolution exists for which only one state has 
non-vanishing geometric phase. Another consequence of the sum rule is that geometric 
phase implementations of qudit gates $\ket{\phi_j} \mapsto e^{i\gamma_g^{(j)}} \ket{\phi_j}$ 
have determinant $+1$, i.e., must be SU($d$). Notably, this excludes geometric 
implementations of Hadamard gates 
\cite{wang21}
\begin{eqnarray}
{\rm H}_d=\frac{1}{\sqrt{d}} \sum_{\mu,\nu=0}^{d-1} \left(e^{i\frac{2\pi}{d}}\right)^{\mu \nu} 
\ket{\mu} \bra{\nu} 
\end{eqnarray}
in certain dimensions $d$. For instance, one can explicitly verify that this is the case for 
$d=2,3,4$, since $\det {\rm H}_2 = -1$ and $\det {\rm H}_3 = \det {\rm H}_4 = -i$.  

To conclude, we have demonstrated that the sum of non-adiabatic geometric phases and the 
underlying curvature 2-forms for a complete set of states always vanish. This finding 
demonstrates that the sum rule holds irrespective of whether the evolution is adiabatic or not. 
It has implications for qudit computation as it limits the types of gates that can be implemented 
by purely geometric means.

\end{document}